\documentclass[aps,prl,reprint,superscriptaddress,longbibliography]{revtex4-2}

\usepackage[T1]{fontenc}
\usepackage{lmodern}
\usepackage{amsmath,amssymb,bm,amsthm}
\usepackage{graphicx}
\usepackage{booktabs}
\usepackage{tabularx}
\usepackage{placeins}
\usepackage{microtype}
\usepackage[hidelinks]{hyperref}
\usepackage{xcolor}

\newcommand{\Tr}{\operatorname{Tr}}
\newcommand{\E}{\mathbb{E}}
\newcommand{\ii}{\mathrm{i}}
\newcommand{\Chat}{\widehat{G}}

\newtheorem{theorem}{Theorem}[section]
\newtheorem{lemma}[theorem]{Lemma}
\newtheorem{proposition}[theorem]{Proposition}

\begin{document}

\title{Postselection-Free Reconstruction of Monitored SPT Flux-Charge Responses}

\author{Shuai Zeng}
\affiliation{School of Communication and Information Engineering, Chongqing University of Posts and Telecommunications, Chongqing 400065, China}
\email{zengshuai@cqupt.edu.cn}

\date{August 15, 2026}

\begin{abstract}
A class of monitored topological invariants compares ordinary and symmetry-twisted quantum trajectories with identical stochastic records, so direct estimation requires exponentially costly matched-record postselection. We replace the twisted quantum experiment by uniformly randomized symmetry eigenstates with known charges, ordinary-boundary monitoring, and an offline twist decoder. Their correlation yields an exact finite-time response identity and removes this experimental sampling bottleneck. When the ordinary charge sharpens and the ordinary-to-flux deformation has trivial net readout-character flow, the response converges to the projective commutator of the symmetry-protected-topological phase; generator-pair responses then determine its finite-Abelian cohomology class. Spectator-sector crossings may close the global Lyapunov gap without changing this character-valued response. In monitored cluster circuits, the decoder is exact and polynomial in the Gaussian limit, while Gaussian-drop reconstruction treats interacting records; independent joint-sector spectroscopy and effect-state diagnostics validate the topological deformation.
\end{abstract}

\maketitle

Continuous measurement is becoming an active resource for organizing and protecting many-body quantum information. Its competition with coherent dynamics produces measurement-induced phases distinguished by how quantum information is encoded across stochastic trajectories~\cite{Skinner2019,Lavasani2021}. Symmetry can further stabilize distinct area-law phases, including measurement-induced symmetry-protected-topological (SPT) phases~\cite{Lavasani2021,MorralYepes2023}. Their experimental characterization calls for observables reconstructed directly from the classical measurement record.

SPT order is invisible to generic local observables. In one dimension, its universal content is carried by a projective symmetry action on the entanglement degrees of freedom~\cite{Schuch2011,PollmannTurner2012,Chen2013}. A symmetry flux $h$ inserted through a ring changes the charge of a second symmetry element $g$ by the projective commutator
\begin{equation}
B_{\omega}(g,h)=\frac{\omega(g,h)}{\omega(h,g)},
\label{eq:projective_commutator}
\end{equation}
which is a standard background-field response of the SPT phase~\cite{ShiozakiRyu2017,Wang2015}. For monitored systems, symmetry and topology can also be organized through Kraus and Lyapunov structures~\cite{XiaoKawabata2026}. Oshima \textit{et al.} introduced a trajectory invariant that compares the dominant symmetry charge obtained from ordinary and $h$-twisted dynamics under the \emph{same} stochastic record~\cite{Oshima2025}.

That definition creates an operational obstruction. Let $s=(r_1,\ldots,r_M)$ denote the complete sequence of stochastic outcomes, and let $\tau_h s$ denote its canonically twisted counterpart. Two independently run instruments retain a useful pair only with probability
\begin{equation}
P_{\mathrm{match}}=\sum_s p_0(s)\,p_h(\tau_h s),
\label{eq:match_probability}
\end{equation}
which generically falls exponentially with the number $M$ of outcomes for an extensive stochastic record~\cite{Oshima2025}. Existing postselection-free approaches can access entanglement dynamics, learn ensemble properties or global charges, and classify monitored phases from measurement records~\cite{Ippoliti2021,McGinley2024,Yu2026}. Here we reconstruct the complementary quantized, group-valued flux-charge response of Eq.~\eqref{eq:projective_commutator} from ordinary trajectories and an offline canonical twist decoder.

Here we replace the second quantum experiment by a randomized-input estimator. Operationally, one samples a known symmetry-charge eigenstate, runs the ordinary instrument, and evaluates its record with an offline canonical twist decoder. Their correlation obeys an exact finite-time identity using ordinary-record sampling. The sampling reduction is universal; the classical decoding complexity is set by the model-specific decoder. Charge sharpening converts the soft posterior response into the ordinary/twisted charge product, and the canonical flux--charge relation identifies its quantized limit with Eq.~\eqref{eq:projective_commutator}. We further show that only nontrivial net spectral flow of the \emph{readout character} modifies the endpoint response: crossings among spectator symmetry sectors leave it unchanged even when the global Lyapunov gap closes. The monitored cluster model realizes this with an exact polynomial $J=0$ decoder and a validated Gaussian-drop approximation for interacting records.

\begin{figure*}[t]
\centering
\includegraphics[width=0.88\textwidth]{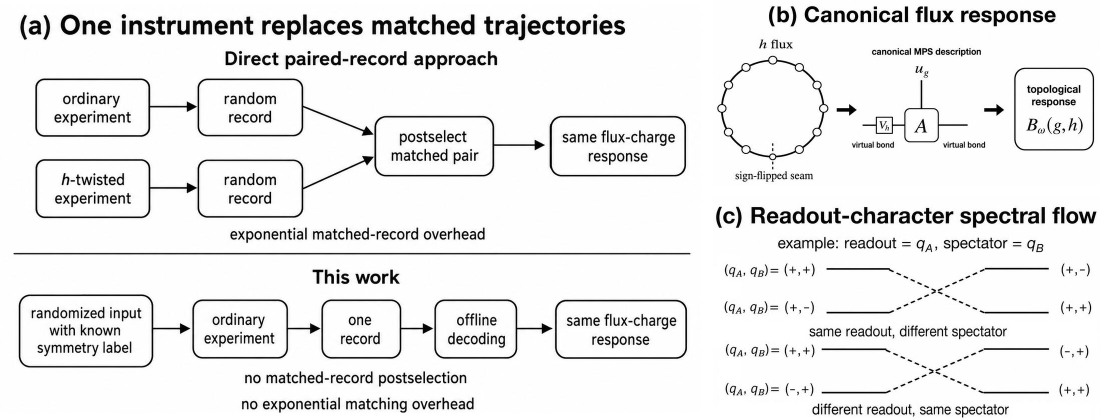}
\caption{\textbf{Single-instrument reconstruction and its topological protection.}
(a) Direct ordinary/$h$-twisted estimation postselects matched record pairs and incurs exponential matching overhead. The randomized protocol uses known-charge inputs, one ordinary record, and offline decoding to reconstruct the same quantized flux--charge response.
(b) A background $h$ flux is gauge fixed to a sign-flipped seam and, in the canonical MPS description, to a virtual $V_h$ insertion carrying $B_\omega(g,h)$.
(c) For readout $q_A$ and spectator $q_B$, spectator exchange leaves the response unchanged, whereas a net readout-character exchange changes it.}
\label{fig:concept}
\end{figure*}

\paragraph{Randomized one-instrument protocol.}
Consider a monitored instrument with ordinary boundary condition $a=0$ and a canonical $h$-twisted counterpart $a=h$. A record $s$ defines a Kraus product $K_s^{(a)}$ and effect operator
\begin{equation}
E_s^{(a)}=K_s^{(a)\dagger}K_s^{(a)},
\qquad
\sum_s E_s^{(0)}=I.
\label{eq:effect}
\end{equation}
Let $G$ be a finite Abelian onsite symmetry group with unitary representation $U_g$. We assume record-wise strong symmetry,
\begin{equation}
[E_s^{(a)},U_g]=0
\quad \text{for all }s,a,g,
\label{eq:strong_symmetry}
\end{equation}
so each effect decomposes into global charge sectors. Choose a complete common symmetry eigenbasis $\{|n\rangle\}_{n=1}^{d}$,
\begin{equation}
U_g|n\rangle=\chi_n(g)|n\rangle.
\label{eq:input_character}
\end{equation}
For the cluster realization below, these inputs may be $X$-basis product states with known $U_A$ and $U_B$ parities.
For a record $s$, let $D_{g,h}^{\rm tr}(s)\in U(1)$ denote the $g$-character value with maximal grouped trace weight in the canonically $h$-twisted effect, using a fixed tie-breaking convention. A practical approximate decoder is denoted by $\widetilde D_{g,h}(s)$.

Each experimental shot consists of five steps: sample $|n\rangle$ uniformly from the common symmetry eigenbasis; run only the ordinary monitored instrument; store the known character $\chi_n(g)$ and the ordinary record $s$; evaluate $D_{g,h}^{\rm tr}(s)$ offline; and return
\begin{equation}
X_{g,h}(n,s)=\chi_n(g)^*D_{g,h}^{\rm tr}(s).
\label{eq:single_shot}
\end{equation}
All quantum operations use the ordinary boundary condition, while the canonical twist enters through a deterministic offline map from the observed record to a symmetry-charge label. Quantum sampling therefore uses a single ordinary record per shot; the classical decoding complexity is analyzed separately below, together with the approximation $\widetilde D_{g,h}$.

The conditional record probability is $p(s|n)=\langle n|E_s^{(0)}|n\rangle$. Uniformly averaging Eq.~\eqref{eq:single_shot} therefore gives
\begin{align}
W_{g,h}
&\equiv \E_{n,s}\!\left[\chi_n(g)^*D_{g,h}^{\rm tr}(s)\right] \notag\\
&=\frac{1}{d}\sum_s D_{g,h}^{\rm tr}(s)\Tr\!\left[U_g^\dagger E_s^{(0)}\right] \notag\\
&=\sum_s p_0(s)m_g^{(0)}(s)D_{g,h}^{\rm tr}(s),
\label{eq:exact_identity}
\end{align}
where
\begin{equation}
p_0(s)=\frac{\Tr E_s^{(0)}}{d},
\qquad
m_g^{(0)}(s)=
\frac{\Tr[U_g^\dagger E_s^{(0)}]}{\Tr E_s^{(0)}}.
\label{eq:posterior_moment}
\end{equation}
Equation~\eqref{eq:exact_identity} follows directly from basis completeness and the Born rule at arbitrary finite size and monitoring time. The quantity $m_g^{(0)}(s)$ is the posterior symmetry-charge moment conditioned on the ordinary record. At finite time its modulus can be smaller than unity, so $W_{g,h}$ is an exact \emph{soft} counterfactual response; charge sharpening and flux protection below yield its quantized limit.

Because $|X_{g,h}|\leq 1$, standard concentration bounds imply that additive precision $\epsilon$ with failure probability $\delta$ requires
\begin{equation}
N=O\!\left(\epsilon^{-2}\ln\delta^{-1}\right)
\label{eq:sampling_complexity}
\end{equation}
independent ordinary trajectories. This removes the exponential experimental repetition associated with Eq.~\eqref{eq:match_probability}. The sampling cost is universal; the classical cost is set by the model-dependent twist decoder. Below we derive the quantization conditions and then construct the exact Gaussian-limit decoder and its interacting Gaussian-drop approximation.

\paragraph{Canonical flux and the SPT response.}
We now make the counterfactual twist microscopic. Consider an even periodic cluster chain with two-site unit cells $(A_x,B_x)$ and strong symmetries
\begin{equation}
U_A=\prod_x X_{A_x},\qquad U_B=\prod_x X_{B_x}.
\label{eq:sublattice_symmetries}
\end{equation}
The cluster generators are
\begin{equation}
K_{A,x}=Z_{B,x-1}X_{A,x}Z_{B,x},\qquad
K_{B,x}=Z_{A,x}X_{B,x}Z_{A,x+1}.
\label{eq:cluster_generators}
\end{equation}
A flat $U_B$ background field with nontrivial holonomy through one seam can be gauge fixed so that only the seam-crossing generator changes sign, $K_{A,0}^{(h)}=-K_{A,0}^{(0)}$ [Fig.~\ref{fig:concept}(b)]. For the weak-measurement filter $M_r(P)=aI+rbP$, one has $M_r(-P)=M_{-r}(P)$; hence the canonical flux is implemented in the offline decoder by flipping the corresponding recorded outcome. The $X_i$ filters and the symmetry-preserving $X_iX_{i+2}$ interactions are neutral under this background field and remain unchanged. This identifies the decoder operation as the lattice realization of a background-gauge holonomy.

For a normal symmetric matrix-product-state branch with virtual action $V_g$,
\begin{equation}
\sum_j [u_g]_{ij}A^j=e^{\ii\theta_g}V_g^\dagger A^iV_g,
\qquad
V_gV_h=\omega(g,h)V_{gh},
\label{eq:mps_symmetry}
\end{equation}
the canonical MPS $h$-flux state is defined by inserting $V_h$ on the closing virtual bond. Moving the $g$ action through this defect gives
\begin{equation}
D_{g,0}^{\rm eig*}D_{g,h}^{\rm eig}
=B_\omega(g,h)
=\frac{\omega(g,h)}{\omega(h,g)}.
\label{eq:mps_flux_charge}
\end{equation}
This standard MPS relation identifies the cohomological meaning of the canonical flux response. The background-field prescription fixes the canonical defect: a zero-dimensional seam decoration would multiply the flux charge by an ordinary linear character while leaving the bulk SPT class unchanged, so it represents a distinct defect choice. For the monitored effects studied here, the explicit lattice deformations and readout-character flow below establish the microscopic identification with the canonical flux. The virtual-space proof and canonical-defect conditions are given in the Supplemental Material~\cite{SupplementalMaterial}.

To connect the physical effect to controlled endpoints, we construct two explicit record-wise deformations. Along the \emph{cluster path}, the $X_i$ measurements and $X_iX_{i+2}$ interactions are continuously switched off while an arbitrarily weak, symmetry-preserving cluster pinning removes accidental endpoint degeneracies. At the commuting endpoint, ordinary and twisted effects have the exact product
\begin{equation}
B_{\rm cluster}=-1.
\label{eq:cluster_endpoint}
\end{equation}
Along the \emph{trivial path}, cluster measurements and interactions are switched off while an arbitrarily weak $X$ pinning is introduced. The flux is then invisible to all remaining filters, so the two endpoint effects coincide and
\begin{equation}
B_X=+1.
\label{eq:trivial_endpoint}
\end{equation}
The proof path includes the pinning for $\lambda>0$ and returns exactly to the physical instrument at $\lambda=0$. These deformations connect the physical effects to explicit fixed points and reduce quantization to a directly testable character-flow condition.

\paragraph{Character-resolved spectral protection.}
The relevant spectral condition is weaker than a globally open Lyapunov gap. Decompose an effect along a continuous path as
\begin{equation}
E^{(a)}(\lambda)=\bigoplus_{\alpha\in\Chat}E_{\alpha}^{(a)}(\lambda),
\qquad
\Lambda_\alpha^{(a)}(\lambda)=\lambda_{\max}\!\left(E_\alpha^{(a)}(\lambda)\right).
\label{eq:character_blocks}
\end{equation}
For a chosen readout element $g$, group all irreducible sectors with the same character value,
\begin{equation}
\begin{aligned}
\Lambda_z^{(a,g)}(\lambda)
&=\max_{\alpha:\,\alpha(g)=z}\Lambda_\alpha^{(a)}(\lambda),\\
D_{g,a}^{\rm eig}(\lambda)
&=\arg\max_z\Lambda_z^{(a,g)}(\lambda).
\end{aligned}
\label{eq:grouped_character_weights}
\end{equation}
Continuity then gives a readout-character spectral-flow theorem: $D_{g,a}^{\rm eig}$ can change only when grouped weights with distinct values $z\neq z'$ cross. For a $\mathbb Z_2$ readout,
\begin{equation}
D_{g,a}^{\rm eig}(0)=D_{g,a}^{\rm eig}(1)(-1)^{N_{g,a}},
\label{eq:flow_parity}
\end{equation}
where $N_{g,a}$ counts isolated crossings that change the measured character, modulo two. More generally, the relevant event is that the combined ordinary-plus-twisted endpoint product differs from the fixed-point value. Thus the response changes only under net readout-character flow.

Crossings between sectors $\alpha$ and $\alpha'$ satisfying $\alpha(g)=\alpha'(g)$ are spectators: they may exchange the globally dominant eigenstate without changing the decoded response [Fig.~\ref{fig:concept}(c), with readout $q_A$ and spectator $q_B$]. This produces a hierarchy of three gaps [Fig.~\ref{fig:validation}(a)]. The \emph{readout-character gap} separates the grouped maxima in Eq.~\eqref{eq:grouped_character_weights} and locally protects $D_g^{\rm eig}$; a \emph{spectator gap} inside the winning character group may close without changing the response; and the top-two gap inside a complete symmetry block checks continuity of the many-body state after all global charges are fixed. In the interacting cluster trajectories, phase-matched deformations preserve the fixed-point endpoint product and a finite complete-block internal gap even though the ordinary and twisted branch characters can switch along the path. Only their net contribution to the endpoint product matters. The resolved spectator crossing also retains the appropriate string-order and entanglement-spectrum diagnostics (Supplemental Material~\cite{SupplementalMaterial}). Phase-mismatched fixed-point paths provide controls with nontrivial net readout-character flow [Fig.~\ref{fig:validation}(b)]. Consequently, a character-valued topological response can remain quantized even when the globally smallest Lyapunov gap closes.

Finally, group the ordinary posterior weights by their value on the measured element $g$, and let $\eta_0^{(g)}(s)$ be the weight outside the winning $g$-character. For the $\mathbb Z_2$ response used below, let $P_{\rm odd}^{(g)}$ be the probability that the ordinary/twisted endpoint product differs from its fixed-point value (equivalently, an odd combined flow parity for isolated crossings, with endpoint ties included under the fixed tie-breaking convention), let $P_{\rm tr/eig}$ denote a trace/eigen character disagreement on either branch, and let $P_{\rm dec}$ denote disagreement between $\widetilde D_{g,h}$ and the exact trace-character decoder. Then
\begin{equation}
|\widetilde W-B_\omega|
\leq
2\E\eta_0^{(g)}
+2P_{\rm odd}^{(g)}
+2P_{\rm tr/eig}
+2P_{\rm dec}.
\label{eq:error_budget}
\end{equation}
The bound is insensitive to spectator crossings because they preserve the measured character. For a general finite-Abelian response, the term $2P_{\rm odd}^{(g)}$ is replaced by the mean character-flow distance $\E|D_{g,0}^{\rm eig*}D_{g,h}^{\rm eig}-B_\omega|$. Equation~\eqref{eq:error_budget} turns the asymptotic topology statement into a finite-time bound with separately quantifiable contributions. Full interacting-effect calculations resolve all four terms; their deep-phase empirical sums are $0.179$ and $0.352$ at $L=8$ and are smaller at the larger system sizes (Supplemental Material~\cite{SupplementalMaterial}). The cohomology label is unique whenever the bound is smaller than half the separation between distinct allowed bicharacter values; for a $\mathbb Z_2$ response, an error below unity fixes the sign.

\paragraph{Interacting reconstruction.}
The monitored cluster circuit admits an exact polynomial decoder in the Gaussian limit and a polynomial Gaussian-drop approximation for interacting records. Each local update weakly measures either $X_i$ with probability $p_X$ or the cluster operator $K_i$ with probability $1-p_X$. At $J=0$, an $X$-diagonal Jordan--Wigner map makes the measurement filters Gaussian after resolving total fermion parity; exact parity projection and recombination then give $D_{A,h}^{\rm tr}(s)$ for every finite record in $O(ML^2+L^3)$ time and $O(L^2)$ memory, without a Lyapunov-gap assumption (Supplemental Material~\cite{SupplementalMaterial}). To probe the non-Gaussian regime, we generate Born records with symmetry-preserving interactions $e^{-\ii JX_iX_{i+2}}$ at rate $r_J$, using $\beta=2$, $J=0.1$, and $r_J=0.2$. The same covariance/Pfaffian construction becomes the Gaussian-drop decoder by omitting only these interaction gates from decoder propagation.

Independent full-effect calculations validate the topological bridge. Twenty $L=8$ records per deep phase give the correct fixed-point endpoint response on all forty phase-matched paths. On the nine-point deformation grid, individual $q_A$ branch switches occur in $9/20$ SPT and $8/20$ trivial records, yet the endpoint product remains fixed; one SPT record shows a transient product reversal before returning. Spectator $U_B$ switching occurs in $11/20$ and $6/20$ records, while the complete-sector internal gap remains resolved. These calculations resolve the net character flow encoded by the endpoint product while allowing intermediate branch crossings. Phase-mismatched controls exhibit nontrivial net flow, and near the transition reduced character gaps and mixed endpoint labels mark the loss of phase-resolved protection (Supplemental Material~\cite{SupplementalMaterial}).

Figure~\ref{fig:validation}(c) reports the finite-size performance for $(L,T)=(8,8),(10,12),(12,16)$ with 100 interacting records per phase. Additional 100-record scans across $p_X$ and monitoring time resolve the crossover and finite-time sharpening (Supplemental Material~\cite{SupplementalMaterial}). The Gaussian-drop responses at the deep-phase cuts are $-0.86,-0.96,-0.96$ in the SPT regime and $+0.90,+0.96,+1.00$ in the trivial regime. The error bars are exact 95\% Clopper--Pearson intervals inferred from the corresponding binary failure counts; for zero observed failures they give the corresponding finite upper confidence bound. The decoder-failure fractions are $(7,2,2)\%$ and $(5,2,0)\%$, while independent full interacting-effect calculations observe readout-character endpoint mismatches at $0$--$3\%$; these include any gap-closing tie under the fixed tie-breaking convention, with finite-sample intervals reported in the Supplemental Material~\cite{SupplementalMaterial}. Figure~\ref{fig:validation}(d) benchmarks the decoder for $M=L(2L-8)$ up to $L=32$: the median time remains below $0.21$ s and follows an effective $L^{3.14}$ dependence over the tested range, consistent with the analytical polynomial complexity.

\begin{figure*}[t]
\centering
\includegraphics[width=0.92\textwidth]{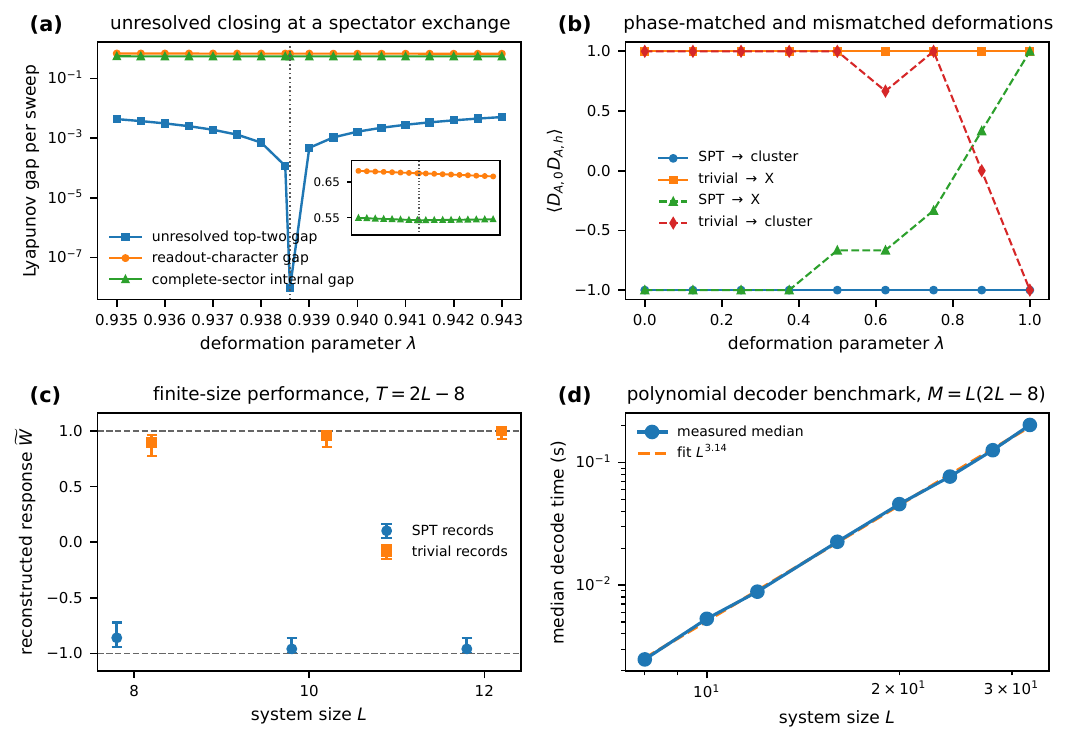}
\caption{\textbf{Character-resolved protection and interacting reconstruction.}
(a) Selected $L=8$ SPT record: the unresolved top-two gap closes at a spectator-charge exchange while the readout-character and complete-sector internal gaps remain finite. Inset: the two finite gaps on a linear scale.
(b) Phase-matched deformations preserve the fixed-point response; mismatched paths require nontrivial net readout-character flow.
(c) Reconstructed response for $T=2L-8$ with 100 interacting records per phase. Error bars are exact 95\% Clopper--Pearson intervals; dashed lines mark the quantized values.
(d) Median decoding time for $M=L(2L-8)$; the dashed line is the fitted $L^{3.14}$ dependence.}
\label{fig:validation}
\end{figure*}

For
\begin{equation}
G=\prod_{i=1}^{r}\mathbb Z_{N_i},
\label{eq:finite_abelian_group}
\end{equation}
the generator-pair responses $B_\omega(g_i,g_j)$ determine the alternating bicharacter and hence the full cohomology class $[\omega]\in H^2(G,U(1))$. Each pair takes values among the $\gcd(N_i,N_j)$-th roots of unity, so at most $r(r-1)/2$ response settings reconstruct the class, with trivial factors omitted. The protocol therefore reconstructs the group-valued topological invariant from generator-pair responses.

In summary, randomized symmetry-charge inputs replace the matched-counterfactual-trajectory comparison by a one-instrument randomized estimator. The exact identity separates sampling from topology; the canonical flux--charge relation and character-flow theorem identify the quantized regime; and the cluster model supplies an exact Gaussian-limit decoder with validated Gaussian-drop reconstruction for interacting records. A global Lyapunov gap may close through spectator sectors while the selected character-valued response remains intact. These results establish a single-instrument route to quantized flux-charge responses in monitored quantum matter.

\paragraph*{Data availability.---} The code used to generate and validate the numerical results is included as ancillary material with this arXiv submission; the reported numerical data can be regenerated from the specified parameters and random seeds.

\bibliography{references}

@article{Skinner2019,
  author = {Skinner, Brian and Ruhman, Jonathan and Nahum, Adam},
  title = {Measurement-Induced Phase Transitions in the Dynamics of Entanglement},
  journal = {Phys. Rev. X},
  volume = {9},
  pages = {031009},
  year = {2019},
  doi = {10.1103/PhysRevX.9.031009}
}

@article{Lavasani2021,
  author = {Lavasani, Ali and Alavirad, Yahya and Barkeshli, Maissam},
  title = {Measurement-induced topological entanglement transitions in symmetric random quantum circuits},
  journal = {Nat. Phys.},
  volume = {17},
  pages = {342--347},
  year = {2021},
  doi = {10.1038/s41567-020-01112-z}
}

@article{MorralYepes2023,
  author = {Morral-Yepes, Ra{\'u}l and Pollmann, Frank and Lovas, Izabella},
  title = {Detecting and stabilizing measurement-induced symmetry-protected topological phases in generalized cluster models},
  journal = {Phys. Rev. B},
  volume = {108},
  pages = {224304},
  year = {2023},
  doi = {10.1103/PhysRevB.108.224304}
}

@article{Oshima2025,
  author = {Oshima, H. and Mochizuki, K. and Hamazaki, R. and Fuji, Y.},
  title = {Topology and Spectrum in Measurement-Induced Phase Transitions},
  journal = {Phys. Rev. Lett.},
  volume = {134},
  pages = {240401},
  year = {2025},
  doi = {10.1103/PhysRevLett.134.240401}
}

@article{PollmannTurner2012,
  author = {Pollmann, Frank and Turner, Ari M.},
  title = {Detection of symmetry-protected topological phases in one dimension},
  journal = {Phys. Rev. B},
  volume = {86},
  pages = {125441},
  year = {2012},
  doi = {10.1103/PhysRevB.86.125441}
}

@article{Schuch2011,
  author = {Schuch, Norbert and P{\'e}rez-Garc{\'i}a, David and Cirac, Ignacio},
  title = {Classifying quantum phases using matrix product states and projected entangled pair states},
  journal = {Phys. Rev. B},
  volume = {84},
  pages = {165139},
  year = {2011},
  doi = {10.1103/PhysRevB.84.165139}
}

@article{Chen2013,
  author = {Chen, Xie and Gu, Zheng-Cheng and Liu, Zheng-Xin and Wen, Xiao-Gang},
  title = {Symmetry protected topological orders and the group cohomology of their symmetry group},
  journal = {Phys. Rev. B},
  volume = {87},
  pages = {155114},
  year = {2013},
  doi = {10.1103/PhysRevB.87.155114}
}

@article{ShiozakiRyu2017,
  author = {Shiozaki, Ken and Ryu, Shinsei},
  title = {Matrix product states and equivariant topological field theories for bosonic symmetry-protected topological phases in (1+1) dimensions},
  journal = {J. High Energy Phys.},
  volume = {2017},
  number = {04},
  pages = {100},
  year = {2017},
  doi = {10.1007/JHEP04(2017)100}
}

@article{Wang2015,
  author = {Wang, Juven and Gu, Zheng-Cheng and Wen, Xiao-Gang},
  title = {Field-Theory Representation of Gauge-Gravity Symmetry-Protected Topological Invariants, Group Cohomology, and Beyond},
  journal = {Phys. Rev. Lett.},
  volume = {114},
  pages = {031601},
  year = {2015},
  doi = {10.1103/PhysRevLett.114.031601}
}

@misc{Yu2026,
  author = {Yu, Hui and Hu, Jiangping and Zhang, Shi-Xin},
  title = {Post-Selection-Free Decoding of Measurement-Induced Area-Law Phases via Neural Networks},
  eprint = {2604.03550},
  archivePrefix = {arXiv},
  year = {2026}
}

@article{McGinley2024,
  author = {McGinley, Max},
  title = {Postselection-free learning of measurement-induced quantum dynamics},
  journal = {PRX Quantum},
  volume = {5},
  pages = {020347},
  year = {2024},
  doi = {10.1103/PRXQuantum.5.020347}
}

@article{Ippoliti2021,
  author = {Ippoliti, Matteo and Khemani, Vedika},
  title = {Postselection-Free Entanglement Dynamics via Spacetime Duality},
  journal = {Phys. Rev. Lett.},
  volume = {126},
  pages = {060501},
  year = {2021},
  doi = {10.1103/PhysRevLett.126.060501}
}

@article{XiaoKawabata2026,
  author = {Xiao, Zhenyu and Kawabata, Kohei},
  title = {Symmetry and Topology of Monitored Quantum Dynamics},
  journal = {Phys. Rev. B},
  volume = {113},
  pages = {134307},
  year = {2026},
  doi = {10.1103/1q83-jcbw}
}

@misc{SupplementalMaterial,
  note = {See Supplemental Material appended to this preprint for complete finite-time proofs, the exact Gaussian decoder, numerical validation, and implementation details.}
}

\clearpage
\onecolumngrid
\setcounter{section}{0}
\setcounter{subsection}{0}
\setcounter{equation}{0}
\setcounter{table}{0}
\setcounter{figure}{0}
\renewcommand{\thesection}{S\arabic{section}}
\renewcommand{\thesubsection}{\thesection.\arabic{subsection}}
\renewcommand{\theequation}{S\arabic{equation}}
\renewcommand{\thetable}{S\arabic{table}}
\renewcommand{\thefigure}{S\arabic{figure}}
\renewcommand{\theHsection}{supp.\arabic{section}}
\renewcommand{\theHsubsection}{supp.\arabic{section}.\arabic{subsection}}
\renewcommand{\theHequation}{supp.\arabic{equation}}
\renewcommand{\theHtable}{supp.\arabic{table}}
\renewcommand{\theHfigure}{supp.\arabic{figure}}

\begin{center}
{\large\bfseries Supplemental Material for ``Postselection-Free Reconstruction of Monitored SPT Flux-Charge Responses''}\\[0.8em]
Shuai Zeng\\
{\small School of Communication and Information Engineering, Chongqing University of Posts and Telecommunications, Chongqing 400065, China}\\
{\small \texttt{zengshuai@cqupt.edu.cn}}\\[0.8em]
\end{center}

This Supplemental Material gives the complete finite-time response proof, charge-sharpening and sampling bounds, the microscopic flux construction, the canonical matrix-product-state (MPS) flux-charge relation, explicit fixed-point deformations, the readout-character spectral-flow theorem, the trace/eigen-sector bridge, the full error budget, the finite-Abelian cohomology reconstruction, the exact Gaussian trace-character decoder, and numerical and implementation details.

\section{General randomized-response identity}

Let $\mathcal H$ have dimension $d$, and let the ordinary monitored instrument have Kraus products $K_s^{(0)}$ indexed by complete records $s$. The corresponding effects are
\begin{equation}
E_s^{(0)}=K_s^{(0)\dagger}K_s^{(0)},\qquad
E_s^{(0)}\ge 0,\qquad
\sum_sE_s^{(0)}=I.
\label{eq:s_effects}
\end{equation}
Let $G$ be a finite Abelian onsite symmetry group represented by unitaries $U_g$. We assume record-wise strong symmetry,
\begin{equation}
[E_s^{(a)},U_g]=0
\quad\text{for every record }s,\text{ boundary condition }a,\text{ and }g\in G.
\label{eq:s_strong}
\end{equation}
Choose a complete common eigenbasis $\{|n\rangle\}_{n=1}^{d}$,
\begin{equation}
U_g|n\rangle=\chi_n(g)|n\rangle.
\label{eq:s_inputchar}
\end{equation}
For each record, $D_{g,h}^{\rm tr}(s)\in U(1)$ denotes the $g$-character value with maximal grouped trace weight in the canonically $h$-twisted effect, using a fixed tie-breaking convention. The practical approximate decoder is denoted $\widetilde D_{g,h}(s)$.

\begin{theorem}[Finite-time randomized-response identity]
If $|n\rangle$ is sampled uniformly, the ordinary instrument is run, and the single-shot variable
\begin{equation}
X_{g,h}(n,s)=\chi_n(g)^*D_{g,h}^{\rm tr}(s)
\label{eq:s_single}
\end{equation}
is recorded, then
\begin{equation}
W_{g,h}
\equiv
\E_{n,s}X_{g,h}
=
\frac{1}{d}\sum_sD_{g,h}^{\rm tr}(s)\Tr[U_g^\dagger E_s^{(0)}]
=
\sum_sp_0(s)m_g^{(0)}(s)D_{g,h}^{\rm tr}(s)
\label{eq:s_identity}
\end{equation}
for every finite system size and monitoring time, where
\begin{equation}
p_0(s)=\frac{\Tr E_s^{(0)}}{d},
\qquad
m_g^{(0)}(s)=
\frac{\Tr[U_g^\dagger E_s^{(0)}]}{\Tr E_s^{(0)}}.
\label{eq:s_moment}
\end{equation}
\end{theorem}

\begin{proof}
The conditional probability of record $s$ given input $|n\rangle$ is
\begin{equation}
p(s|n)=\langle n|E_s^{(0)}|n\rangle.
\end{equation}
Therefore
\begin{align}
\E_{n,s}X_{g,h}
&=
\frac1d\sum_{n,s}\chi_n(g)^*D_{g,h}^{\rm tr}(s)
\langle n|E_s^{(0)}|n\rangle \notag\\
&=
\frac1d\sum_sD_{g,h}^{\rm tr}(s)
\sum_n\langle n|U_g^\dagger E_s^{(0)}|n\rangle \notag\\
&=
\frac1d\sum_sD_{g,h}^{\rm tr}(s)\Tr[U_g^\dagger E_s^{(0)}],
\end{align}
which is Eq.~\eqref{eq:s_identity}.
\end{proof}

The exact identity follows from randomized experimental design and the Born rule at arbitrary finite size and monitoring time. It is written for $D_{g,h}^{\rm tr}$; replacing this quantity by $\widetilde D_{g,h}$ produces the decoder contribution treated in the finite-time error budget below.

The construction extends to a general ensemble $\{p_n,\rho_n,w_n\}$ and yields the same response whenever
\begin{equation}
\sum_np_nw_n\rho_n=\frac{U_g^\dagger}{d}.
\label{eq:s_generalensemble}
\end{equation}
The eigenbasis construction is operationally simple because $w_n=\chi_n(g)^*$ has unit modulus and the input label is known before the monitored evolution.

\section{Charge sharpening and sampling complexity}

Because $G$ is finite Abelian, its irreducible representations are one-dimensional characters $\alpha\in\Chat$. Define
\begin{equation}
P_\alpha=\frac1{|G|}\sum_{x\in G}\alpha(x)^*U_x,
\qquad
Z_\alpha^{(a)}(s)=\Tr[P_\alpha E_s^{(a)}].
\label{eq:s_projectors}
\end{equation}
For the ordinary effect,
\begin{equation}
p_\alpha(s)=\frac{Z_\alpha^{(0)}(s)}{\Tr E_s^{(0)}},
\qquad
\sum_\alpha p_\alpha(s)=1.
\label{eq:s_palpha}
\end{equation}
For a fixed readout element $g$, group irreducible sectors by the character value $z=\alpha(g)$:
\begin{equation}
p_z^{(g)}(s)=\sum_{\alpha:\alpha(g)=z}p_\alpha(s),
\qquad
m_g^{(0)}(s)=\sum_zp_z^{(g)}(s)z^*.
\label{eq:s_grouped_posterior}
\end{equation}
Let $z_0(s)$ uniquely maximize $p_z^{(g)}(s)$ and define
\begin{equation}
\eta_0^{(g)}(s)=1-p_{z_0}^{(g)}(s),
\qquad
D_{g,0}^{\rm tr}(s)=z_0(s).
\label{eq:s_grouped_eta}
\end{equation}

\begin{lemma}[Readout-character sharpening bound]
\begin{equation}
|m_g^{(0)}(s)-D_{g,0}^{\rm tr}(s)^*|
\le 2\eta_0^{(g)}(s).
\label{eq:s_sharp}
\end{equation}
\end{lemma}

\begin{proof}
Using Eq.~\eqref{eq:s_grouped_posterior},
\begin{align}
m_g^{(0)}-z_0^*
&=
[p_{z_0}^{(g)}-1]z_0^*
+\sum_{z\ne z_0}p_z^{(g)}z^*.
\end{align}
Every character value has unit modulus, so the triangle inequality gives
\begin{equation}
|m_g^{(0)}-z_0^*|
\le (1-p_{z_0}^{(g)})+\sum_{z\ne z_0}p_z^{(g)}
=2\eta_0^{(g)}.
\end{equation}
\end{proof}

For a $\mathbb Z_2$ readout, $m_g^{(0)}\in[-1,1]$ and
\begin{equation}
2\eta_0^{(g)}=1-|m_g^{(0)}|.
\label{eq:s_z2exact}
\end{equation}
Thus the finite-time soft response becomes the hard readout-character product continuously as the ordinary charge sharpens. Spectator sectors that share the same value on $g$ are already combined in Eq.~\eqref{eq:s_grouped_posterior}.

A finite-time sufficient condition can be stated directly in terms of effect eigenvalues. Let $\Lambda_*$ be the largest eigenvalue in the winning readout-character group and suppose every eigenvalue outside that group is at most $\Lambda_*e^{-2T\Delta_c}$. If the total dimension of the competing groups is $d_{\rm out}$, then
\begin{equation}
\eta_0^{(g)}(s)
\le d_{\rm out}e^{-2T\Delta_c}.
\label{eq:s_lyapsharp}
\end{equation}
Using $d_{\rm out}<q^L$ gives the sufficient joint size-time condition
\begin{equation}
2T\Delta_c-L\log q\longrightarrow+\infty.
\label{eq:s_sufficient_joint}
\end{equation}

Because $|X_{g,h}|=1$, ordinary bounded-variable concentration applies. For a real $\mathbb Z_2$ response, Hoeffding's inequality yields
\begin{equation}
\Pr(|\widehat W-W|\ge\epsilon)
\le 2e^{-N\epsilon^2/2}.
\label{eq:s_hoeffdingreal}
\end{equation}
For a complex response, applying the real bound separately to real and imaginary parts gives, for example,
\begin{equation}
\Pr(|\widehat W-W|\ge\epsilon)
\le 4e^{-N\epsilon^2/4}.
\label{eq:s_hoeffdingcomplex}
\end{equation}
Hence fixed additive precision and confidence require
\begin{equation}
N=O\!\left(\epsilon^{-2}\log\delta^{-1}\right)
\label{eq:s_samplecomplexity}
\end{equation}
ordinary trajectories. This sampling statement is independent of whether the offline decoder itself is efficient.

\section{Canonical microscopic symmetry twist}

We now specialize to the even-length cluster ring. Group sites into unit cells
\begin{equation}
A_x=2x,\qquad B_x=2x+1,\qquad x=0,\ldots,N-1,
\end{equation}
with $L=2N$. The two onsite symmetries are
\begin{equation}
U_A=\prod_xX_{A_x},
\qquad
U_B=\prod_xX_{B_x}.
\label{eq:s_UAB}
\end{equation}
The cluster generators are
\begin{equation}
K_{A,x}=Z_{B,x-1}X_{A,x}Z_{B,x},
\qquad
K_{B,x}=Z_{A,x}X_{B,x}Z_{A,x+1}.
\label{eq:s_cluster_generators}
\end{equation}
They obey
\begin{equation}
\prod_xK_{A,x}=U_A,\qquad
\prod_xK_{B,x}=U_B.
\label{eq:s_product_stabilizers}
\end{equation}

A two-outcome weak measurement of a Pauli operator $P$ is represented by
\begin{equation}
M_r(P)=
\frac{e^{\beta rP/2}}{\sqrt{2\cosh\beta}}
=aI+rbP,
\qquad r=\pm1,
\label{eq:s_weakfilter}
\end{equation}
where
\begin{equation}
a=\frac{\cosh(\beta/2)}{\sqrt{2\cosh\beta}},
\qquad
b=\frac{\sinh(\beta/2)}{\sqrt{2\cosh\beta}}.
\end{equation}

Insert a flat $U_B$ background connection with holonomy $h=U_B$ through the ring, and choose a one-cut gauge on the bond between unit cells $N-1$ and $0$. Only the seam-crossing $A$-sublattice generator changes:
\begin{equation}
K_{A,0}^{(h)}=-K_{A,0},
\qquad
K_{A,x\ne0}^{(h)}=K_{A,x},
\qquad
K_{B,x}^{(h)}=K_{B,x}.
\label{eq:s_microscopic_twist}
\end{equation}
Equation~\eqref{eq:s_weakfilter} then gives
\begin{equation}
M_r(K_{A,0}^{(h)})
=
M_r(-K_{A,0})
=
M_{-r}(K_{A,0}).
\label{eq:s_outcome_flip}
\end{equation}
Thus the canonical twist is implemented in the decoder by flipping the sign of the recorded outcome whenever the seam-crossing cluster measurement occurs. The one-site $X_i$ filters and the symmetry-preserving interactions $e^{-\ii JX_iX_{i+2}}$ are neutral under this flat connection and require no modification.

Moving the gauge cut conjugates the microscopic operators by a product of onsite $U_B$ transformations on an interval. Because the ordinary and twisted effects are compared with the same gauge convention and $U_A$ commutes with this gauge transformation, the decoded $U_A$ charge and the flux-charge response are cut independent. What matters is the holonomy, not the position of the seam.

\section{Canonical MPS flux-charge relation}

Consider first a translation-invariant injective MPS
\begin{equation}
|\Psi_0\rangle
=
\sum_{\boldsymbol i}
\Tr(A^{i_1}\cdots A^{i_L})
|\boldsymbol i\rangle.
\label{eq:s_mps0}
\end{equation}
Symmetry of the state implies that the physical action can be pushed to the virtual space:
\begin{equation}
\sum_j(u_g)_{ij}A^j
=
e^{\ii\theta_g}V_g^\dagger A^iV_g.
\label{eq:s_mpssym}
\end{equation}
The virtual operators form a projective representation,
\begin{equation}
V_gV_h=\omega(g,h)V_{gh}.
\label{eq:s_projective}
\end{equation}
For Abelian $G$, define the gauge-invariant alternating bicharacter
\begin{equation}
B_\omega(g,h)
=
\frac{\omega(g,h)}{\omega(h,g)}.
\label{eq:s_bichar}
\end{equation}

The canonical $h$-flux state is obtained by inserting $V_h$ on the virtual closing bond:
\begin{equation}
|\Psi_h\rangle
=
\sum_{\boldsymbol i}
\Tr(A^{i_1}\cdots A^{i_L}V_h)
|\boldsymbol i\rangle.
\label{eq:s_mpsflux}
\end{equation}

\begin{theorem}[Canonical MPS flux-charge relation]
With the orientation convention of Eq.~\eqref{eq:s_mpsflux},
\begin{equation}
D_{g,0}^{\rm eig*}D_{g,h}^{\rm eig}=B_\omega(g,h).
\label{eq:s_mps_theorem}
\end{equation}
\end{theorem}

\begin{proof}
Pushing $U_g=u_g^{\otimes L}$ through every tensor in Eq.~\eqref{eq:s_mps0} produces cancelling pairs of $V_g$ and $V_g^\dagger$, leaving
\begin{equation}
U_g|\Psi_0\rangle=e^{\ii L\theta_g}|\Psi_0\rangle.
\end{equation}
For the flux state, the same telescoping leaves
\begin{equation}
V_gV_hV_g^\dagger=B_\omega(g,h)V_h,
\end{equation}
so
\begin{equation}
U_g|\Psi_h\rangle
=
e^{\ii L\theta_g}B_\omega(g,h)|\Psi_h\rangle.
\end{equation}
Taking the ratio of the twisted and ordinary dominant-state $g$ charges gives Eq.~\eqref{eq:s_mps_theorem}.
\end{proof}

The same argument applies to a normal site-dependent MPS with a consistent virtual symmetry class and tensors satisfying
\begin{equation}
\sum_j(u_g)_{ij}A_x^j
=
e^{\ii\theta_{g,x}}
V_{g,x-1}^\dagger A_x^iV_{g,x},
\label{eq:s_nonuniform_mps}
\end{equation}
For such tensors, the virtual matrices telescope around the ring. Provided the flux insertion is nonzero and uses the virtual representation at the closing bond, the ordinary charge is $\prod_xe^{\ii\theta_{g,x}}$ and the canonical flux insertion leaves the same projective commutator. For a record-dependent dominant effect-state MPS, this theorem applies when the microscopic $h$-twisted branch realizes the canonical virtual insertion with the background-field defect choice. In the monitored cluster model below, explicit microscopic fixed-point deformations and readout-character flow establish this identification directly.

\section{Canonical flux selection and defect decoration}

The canonical background flux is the undecorated seam defect. Decorating the seam by a zero-dimensional degree of freedom transforming in a linear representation $\lambda(g)$ preserves the bulk MPS and its cohomology class while multiplying the total flux-state charge:
\begin{equation}
D_{g,0}^{\rm eig*}D_{g,h}^{\rm eig}
=
\lambda(g)B_\omega(g,h).
\label{eq:s_defectdecoration}
\end{equation}
Likewise, a local defect level may cross another level without closing the periodic bulk gap. The selected twisted ground-state or Lyapunov charge can then change while the bulk remains in the same phase.

The canonical flux-charge relation is fixed by (i) the undecorated background-field seam and (ii) spectral continuity of the relevant defect branch. The microscopic holonomy in Sec.~S3 fixes the first condition, while the explicit fixed-point deformations and readout-character theorem below establish the second.

\section{Explicit cluster and trivial fixed-point deformations}

A Born record consists of events
\begin{equation}
s=\{(b_t,i_t,\kappa_t,r_t)\}_{t=1}^{M},
\label{eq:s_record}
\end{equation}
where $b_t$ specifies whether an $X_iX_{i+2}$ interaction was applied, $i_t$ is the measured site, $\kappa_t\in\{X,K\}$ is the measurement type, and $r_t=\pm1$ is the outcome. The physical point is $\lambda=0$.

Along the \emph{cluster path}, we use
\begin{equation}
\beta_X(\lambda)=(1-\lambda)\beta,\qquad
\beta_K(\lambda)=\beta,\qquad
J(\lambda)=(1-\lambda)J,
\label{eq:s_clusterpath}
\end{equation}
and append the symmetry-preserving pinning layer
\begin{equation}
F_{K,\lambda}^{(a)}
=
\prod_i
\exp\!\left[\frac{\lambda\epsilon}{2}K_i^{(a)}\right].
\label{eq:s_clusterpin}
\end{equation}
Along the \emph{trivial path},
\begin{equation}
\beta_X(\lambda)=\beta,\qquad
\beta_K(\lambda)=(1-\lambda)\beta,\qquad
J(\lambda)=(1-\lambda)J,
\label{eq:s_trivialpath}
\end{equation}
with
\begin{equation}
F_{X,\lambda}
=
\prod_i
\exp\!\left[\frac{\lambda\epsilon}{2}X_i\right].
\label{eq:s_xpin}
\end{equation}
The proof path includes this pinning for $\lambda>0$ and coincides exactly with the physical instrument at $\lambda=0$.

At the cluster endpoint, let
\begin{equation}
R_i(s)=\sum_{t:\kappa_t=K,\ i_t=i}r_t.
\end{equation}
All remaining filters commute, and the effect is
\begin{equation}
E_{s,1}^{(a)}
\propto
\exp\!\left[
\sum_i(\beta R_i+\epsilon)K_i^{(a)}
\right].
\label{eq:s_clusterendpoint}
\end{equation}
Choosing $0<\epsilon<\beta$ makes every coefficient nonzero because $R_i$ is an integer. The maximizing eigenvalue of each $K_i^{(a)}$ is therefore fixed by the same sign $\operatorname{sgn}(\beta R_i+\epsilon)$ in the ordinary and twisted problems. The sublattice charge products satisfy
\begin{equation}
U_A=\prod_xK_{A,x}^{(0)}
=
-\prod_xK_{A,x}^{(h)}.
\label{eq:s_endpointminus}
\end{equation}
Hence every record obeys
\begin{equation}
D_{A,0}^{\rm eig}(s,1)D_{A,h}^{\rm eig}(s,1)=-1.
\label{eq:s_clusterB}
\end{equation}

At the trivial endpoint only commuting $X_i$ filters remain. They are insensitive to the $U_B$ holonomy:
\begin{equation}
E_{s,1}^{(h)}=E_{s,1}^{(0)},
\end{equation}
and therefore
\begin{equation}
D_{A,0}^{\rm eig}(s,1)D_{A,h}^{\rm eig}(s,1)=+1.
\label{eq:s_trivialB}
\end{equation}
The two deformations realize the ``defect-compatible path'' explicitly at the level of individual records.

\section{Readout-character spectral-flow theorem}

Let
\begin{equation}
E^{(a)}(\lambda)
=
\bigoplus_{\alpha\in\Chat}E_\alpha^{(a)}(\lambda),
\qquad
\Lambda_\alpha^{(a)}(\lambda)
=
\lambda_{\max}(E_\alpha^{(a)}(\lambda)).
\label{eq:s_irrepblocks}
\end{equation}
For a fixed readout element $g\in G$, sectors that agree on $g$ are operationally indistinguishable to that charge measurement. Define
\begin{equation}
\Lambda_z^{(a,g)}(\lambda)
=
\max_{\alpha:\alpha(g)=z}
\Lambda_\alpha^{(a)}(\lambda),
\qquad
D_{g,a}^{\rm eig}(\lambda)
=
\arg\max_z\Lambda_z^{(a,g)}(\lambda).
\label{eq:s_groupedweights}
\end{equation}

\begin{theorem}[Readout-character spectral flow]
Assume $E^{(a)}(\lambda)$ depends continuously on $\lambda$. The decoded value $D_{g,a}^{\rm eig}$ can change only at a crossing
\begin{equation}
\Lambda_z^{(a,g)}(\lambda_c)
=
\Lambda_{z'}^{(a,g)}(\lambda_c),
\qquad z\ne z'.
\label{eq:s_character_crossing}
\end{equation}
A crossing between irreducible sectors $\alpha,\alpha'$ with $\alpha(g)=\alpha'(g)$ cannot change the $g$-charge response.
\end{theorem}

\begin{proof}
Each $\Lambda_\alpha^{(a)}$ is continuous as the largest eigenvalue of a continuous finite-dimensional Hermitian family. The grouped maximum of finitely many continuous functions is continuous. If one grouped weight is strictly larger than all others on an interval, its label is constant. Therefore a change of $D_{g,a}^{\rm eig}$ requires equality of two grouped weights carrying distinct character values. If $\alpha(g)=\alpha'(g)$, exchanging the winning irreducible sector leaves the grouped label unchanged.
\end{proof}

For a $\mathbb Z_2$ readout with isolated crossings, let $N_{g,a}$ count crossings that change the readout character, modulo two. Then
\begin{equation}
D_{g,a}^{\rm eig}(0)
=
D_{g,a}^{\rm eig}(1)(-1)^{N_{g,a}}.
\label{eq:s_parity}
\end{equation}
Combining ordinary and twisted paths,
\begin{equation}
D_{g,0}^{\rm eig}(0)^*D_{g,h}^{\rm eig}(0)
=
B_{\rm end}
(-1)^{N_{g,0}+N_{g,h}},
\label{eq:s_productparity}
\end{equation}
where $B_{\rm end}=-1$ on the cluster path and $+1$ on the trivial path. More generally, the protection-failure event is defined directly by the endpoint mismatch $D_{g,0}^{\rm eig}(0)^*D_{g,h}^{\rm eig}(0)\ne B_{\rm end}$ under the fixed tie-breaking convention. This definition also includes a readout-character gap closing exactly at an endpoint. Its probability is denoted $P_{\rm odd}^{(g)}$ in the $\mathbb Z_2$ case.

The theorem identifies three distinct finite-size gaps. The \emph{readout-character gap} separates the grouped maxima in Eq.~\eqref{eq:s_groupedweights} and protects the response. A \emph{spectator gap} separates irreducible sectors within the winning character group and may close without changing the selected readout character. Finally, the \emph{complete-sector internal gap} is the top-two gap after all global symmetry labels have been fixed; it diagnoses state continuity inside one full block.

\section{Joint-sector validation and spectator crossing}

For the monitored cluster model, the complete sectors are labeled by $(q_A,q_B)$. Let
\begin{equation}
\Lambda_{q_A,q_B}^{(a)}(\lambda)
\end{equation}
be the largest effect eigenvalue in a joint sector. The response reads $q_A$, so
\begin{equation}
\Lambda_{q_A}^{(a,A)}
=
\max_{q_B=\pm1}\Lambda_{q_A,q_B}^{(a)}.
\label{eq:s_groupqA}
\end{equation}
The three gaps reported numerically are
\begin{align}
\Delta_A^{(a)}
&=
\frac{1}{2T}
\left|
\log\frac{\Lambda_+^{(a,A)}}{\Lambda_-^{(a,A)}}
\right|,
\label{eq:s_gapA}\\
\Delta_B^{\rm spec}
&=
\frac{1}{2T}
\left|
\log\frac{\Lambda_{q_A^*,+}^{(a)}}{\Lambda_{q_A^*,-}^{(a)}}
\right|,
\label{eq:s_gapB}\\
\Delta_{\rm int}^{(a)}
&=
\frac{1}{2T}
\log\frac{\Lambda_{1;q_A^*,q_B^*}^{(a)}}
{\Lambda_{2;q_A^*,q_B^*}^{(a)}}.
\label{eq:s_gapint}
\end{align}

An SPT path exhibiting the spectator exchange resolves the global top-two closing quantitatively. In the ordinary branch, the unresolved top-two gap reaches $9.6\times10^{-9}$ per sweep at $\lambda=0.938599\ldots$, while the readout-character gap is $0.674$ per sweep. In the canonically twisted branch, the corresponding minimum is $1.3\times10^{-8}$ per sweep at $\lambda=0.939343\ldots$, with readout-character gap $0.673$ per sweep. Joint $(q_A,q_B)$ spectroscopy on the same neighborhood shows that the near-degenerate levels have identical $q_A$ and opposite spectator charge $q_B$; on the refined joint grid the complete-sector internal gap remains approximately $0.542$ per sweep. Thus the vanishing global gap is a spectator-sector exchange rather than a change of the readout character.

Independent effect-state diagnostics give the same conclusion. Across the refined ordinary-branch crossing window, the cluster-string magnitude deviates from unity by at most $1.6\times10^{-6}$ and the fourfold entanglement-spectrum flatness by at most $5.8\times10^{-6}$. The dominant effect state therefore remains continuously in the same bulk SPT structure while the spectator sectors exchange order.

Table~\ref{tab:s_jointsummary} summarizes complete phase-matched-path validation at $L=8$ using twenty independent records per deep phase with deterministic seeds. All forty physical-to-fixed-point endpoint products equal the expected response. The nine-point grid detects at least one $q_A$ branch-label change in $9/20$ SPT records and $8/20$ trivial records. Eqs.~\eqref{eq:s_parity}--\eqref{eq:s_productparity} constrain the corresponding net endpoint flow, which remains topologically correct for all forty paths.

\begin{table}[t]
\caption{\textbf{Joint-sector validation on the phase-matched paths.}
``$q_A$ branch flow'' and ``spectator flow'' count records with at least one label change on the nine-point deformation grid. The gap columns give the 10th percentile of the minimum gap over each complete path across twenty records.}
\label{tab:s_jointsummary}
\begin{ruledtabular}
\begin{tabular}{lcccccc}
phase & path & endpoint & $q_A$ branch flow & spectator flow &
$\Delta_A^{p10}$ & $\Delta_{\rm int}^{p10}$\\
\hline
SPT, $p_X=0.25$ & cluster & $20/20$ & $9/20$ & $11/20$ & 0.0246 & 0.250\\
trivial, $p_X=0.75$ & $X$ & $20/20$ & $8/20$ & $6/20$ & 0.0250 & 0.450
\end{tabular}
\end{ruledtabular}
\end{table}

Across the forty phase-matched records summarized in Table~\ref{tab:s_jointsummary}, the endpoint response is always correct and the complete-sector internal gap remains resolved. Branch-character flow is common. In the SPT sample, one record displays a transient reversal of the ordinary--twisted product on the sampled grid and then returns to the correct endpoint value, while the other nineteen keep the product fixed at every sampled point. The trivial sample keeps the product fixed at every sampled point. Spectator switching is also common and can occur at a much smaller gap without affecting the response. These observations directly realize the distinction between local branch crossings, net readout-character flow, and spectator flow in Eqs.~\eqref{eq:s_gapA}--\eqref{eq:s_gapint}.

\section{Trace-character and eigen-character labels}

The randomized identity naturally contains grouped trace weights, whereas the MPS theorem is formulated for the dominant eigenvector. For a fixed readout element $g$, define
\begin{equation}
Z_z^{(a,g)}=\sum_{\alpha:\alpha(g)=z}\Tr E_\alpha^{(a)},
\qquad
\Lambda_z^{(a,g)}=\max_{\alpha:\alpha(g)=z}\lambda_{\max}(E_\alpha^{(a)}).
\label{eq:s_trace_eig_grouped}
\end{equation}
The corresponding labels are
\begin{equation}
D_{g,a}^{\rm tr}=\arg\max_z Z_z^{(a,g)},
\qquad
D_{g,a}^{\rm eig}=\arg\max_z\Lambda_z^{(a,g)}.
\label{eq:s_trace_eig_labels}
\end{equation}
At finite time these labels can differ.

\begin{proposition}[Finite-size trace/eigen character certificate]
Let $z_*$ maximize $\Lambda_z$, and let $d_z$ be the dimension of the complete Hilbert subspace with character value $z$. If
\begin{equation}
\Lambda_{z_*}>d_z\Lambda_z
\qquad\text{for every }z\ne z_*,
\label{eq:s_traceeigencondition}
\end{equation}
then $z_*$ also uniquely maximizes the grouped trace weight. Moreover,
\begin{equation}
\eta_{\rm tr}^{(g)}
\le
\sum_{z\ne z_*}d_z\frac{\Lambda_z}{\Lambda_{z_*}}.
\label{eq:s_traceeigeneta}
\end{equation}
\end{proposition}

\begin{proof}
The winning grouped trace satisfies $Z_{z_*}\ge\Lambda_{z_*}$, whereas $Z_z\le d_z\Lambda_z$. Equation~\eqref{eq:s_traceeigencondition} therefore implies $Z_{z_*}>Z_z$. The posterior weight outside the winning character obeys
\begin{equation}
\eta_{\rm tr}^{(g)}
=
\frac{\sum_{z\ne z_*}Z_z}{Z_{z_*}+\sum_{z\ne z_*}Z_z}
\le
\sum_{z\ne z_*}\frac{Z_z}{Z_{z_*}}
\le
\sum_{z\ne z_*}d_z\frac{\Lambda_z}{\Lambda_{z_*}}.
\end{equation}
\end{proof}

The dimension-times-leading-eigenvalue estimate gives a sufficient certificate. The simulations also measure the empirical probability $P_{\rm tr/eig}$ that the trace and eigen readout characters disagree on at least one of the ordinary and twisted branches.

\section{Finite-time error budget}

Let $\widetilde D_{g,h}$ be the decoder actually used in data processing and define
\begin{equation}
\widetilde W=\E[\chi_n(g)^*\widetilde D_{g,h}(s)].
\end{equation}
Let $B_\omega$ be the fixed-point projective commutator. Introduce the character-flow distance
\begin{equation}
\varepsilon_{\rm flow}^{(g)}
=
\E\left|
D_{g,0}^{\rm eig*}D_{g,h}^{\rm eig}-B_\omega
\right|.
\label{eq:s_flow_distance}
\end{equation}
For a $\mathbb Z_2$ response, $\varepsilon_{\rm flow}^{(g)}=2P_{\rm odd}^{(g)}$.

\begin{theorem}[Finite-time error bound]
For unit-modulus character outputs,
\begin{equation}
|\widetilde W-B_\omega|
\le
2\E\eta_0^{(g)}
+\varepsilon_{\rm flow}^{(g)}
+2P_{\rm tr/eig}
+2P_{\rm dec}.
\label{eq:s_fullerror_general}
\end{equation}
In the $\mathbb Z_2$ model this becomes
\begin{equation}
|\widetilde W-B_\omega|
\le
2\E\eta_0^{(g)}
+2P_{\rm odd}^{(g)}
+2P_{\rm tr/eig}
+2P_{\rm dec}.
\label{eq:s_fullerror}
\end{equation}
Here $P_{\rm dec}$ is the probability that $\widetilde D_{g,h}$ disagrees with the exact trace-character twist decoder, and $P_{\rm tr/eig}$ is the probability of a trace/eigen readout-character disagreement on at least one branch.
\end{theorem}

\begin{proof}
Insert successively the exact trace-decoder response, the hard trace-character product, and the hard eigen-character product:
\begin{align}
|\widetilde W-B_\omega|
&\le
|\widetilde W-W_{\rm tr}|
+
|W_{\rm tr}-\E[D_{g,0}^{\rm tr*}D_{g,h}^{\rm tr}]| \notag\\
&\quad+
|\E[D_{g,0}^{\rm tr*}D_{g,h}^{\rm tr}]
-\E[D_{g,0}^{\rm eig*}D_{g,h}^{\rm eig}]| \notag\\
&\quad+
|\E[D_{g,0}^{\rm eig*}D_{g,h}^{\rm eig}]-B_\omega|.
\label{eq:s_telescoping_error}
\end{align}
The first term is at most $2P_{\rm dec}$. The grouped sharpening lemma bounds the second by $2\E\eta_0^{(g)}$. The third is at most $2P_{\rm tr/eig}$. The final term is at most $\varepsilon_{\rm flow}^{(g)}$ by the triangle inequality, with equality to $2P_{\rm odd}^{(g)}$ for binary outputs. Summing proves Eqs.~\eqref{eq:s_fullerror_general} and \eqref{eq:s_fullerror}.
\end{proof}

The bound is insensitive to spectator crossings because they preserve the selected value of the readout character. For a $\mathbb Z_2$ invariant, any certified total error below one fixes the sign. For a root-of-unity response of order $N$, nearest-root decoding is unique whenever
\begin{equation}
|\widetilde W-B_\omega|<\sin(\pi/N).
\label{eq:s_rootthreshold}
\end{equation}

A direct finite-time comparison at $L=8$ is summarized in Table~\ref{tab:s_soft_hard}. The four-probe stochastic trace estimate provides a numerical consistency check of the analytic randomized-response identity. The soft trace response, hard eigen-character product, and Gaussian-drop response all retain the correct deep-phase sign, while their differences quantify the distinct finite-time contributions.

\begin{table}[t]
\caption{\textbf{Finite-time trace/eigen/decoder comparison at $L=8$.}
The soft response is the four-probe trace estimate $W_{\rm probe}$, the hard column is $\langle D_{g,0}^{\rm eig*}D_{g,h}^{\rm eig}\rangle$, and $\widetilde W$ is the Gaussian-drop response. Uncertainties shown for $W_{\rm probe}$ and $\widetilde W$ are standard errors over 100 interacting records.}
\label{tab:s_soft_hard}
\begin{ruledtabular}
\begin{tabular}{ccccc}
$p_X$ & phase & $W_{\rm probe}$ & hard eigen product & $\widetilde W$\\
\hline
0.25 & SPT & $-0.904\pm0.036$ & $-0.94$ & $-0.86\pm0.051$\\
0.75 & trivial & $+0.888\pm0.033$ & $+0.98$ & $+0.90\pm0.044$
\end{tabular}
\end{ruledtabular}
\end{table}

The four terms in Eq.~\eqref{eq:s_fullerror} can be resolved separately on independent full interacting-effect calculations. Table~\ref{tab:s_errorbudget} reports the resulting empirical decomposition. Here $P_{\rm dec}$ is obtained by comparing the Gaussian-drop twisted label with the exact trace-character twisted label, while $P_{\rm tr/eig}$ isolates trace/eigen disagreement. The entries are finite-sample estimates of the terms in the exact bound.

\begin{table}[t]
\caption{\textbf{Empirical decomposition of the finite-time error budget.}
The last column is the sum of the four preceding contributions in Eq.~\eqref{eq:s_fullerror}.}
\label{tab:s_errorbudget}
\begin{ruledtabular}
\begin{tabular}{ccccccc}
$L$ & phase & $2\langle\eta_0\rangle$ & $2P_{\rm odd}$ & $2P_{\rm tr/eig}$ & $2P_{\rm dec}$ & sum\\
\hline
8 & SPT & 0.039 & 0.060 & 0.040 & 0.040 & 0.179\\
8 & trivial & 0.072 & 0.020 & 0.140 & 0.120 & 0.352\\
10 & SPT & 0.022 & 0 & 0.040 & 0.040 & 0.102\\
10 & trivial & 0.056 & 0 & 0 & 0 & 0.056\\
12 & SPT & $4.3\!\times\!10^{-4}$ & 0 & 0 & 0 & $4.3\!\times\!10^{-4}$\\
12 & trivial & 0.009 & 0 & 0 & 0 & 0.009
\end{tabular}
\end{ruledtabular}
\end{table}

At $L=8$ the empirical sum lies well below the binary sign-separation scale of one in both deep phases, and the larger-size calculations give still smaller sums. These values provide a finite-sample numerical consistency test of the error decomposition; Table~\ref{tab:s_binomial_ci} supplies confidence intervals for the principal binary events.

\section{Finite-Abelian cohomology reconstruction}

Let
\begin{equation}
G=\prod_{i=1}^{r}\mathbb Z_{N_i},
\label{eq:s_groupdecomp}
\end{equation}
with generators $e_i$. For finite Abelian $G$, bosonic one-dimensional SPT classes are classified by alternating bicharacters, and
\begin{equation}
H^2(G,U(1))
\simeq
\prod_{i<j}\mathbb Z_{\gcd(N_i,N_j)}.
\label{eq:s_H2}
\end{equation}
The generator-pair responses satisfy
\begin{equation}
B_\omega(e_i,e_j)^{\gcd(N_i,N_j)}=1,
\qquad
B_\omega(e_j,e_i)=B_\omega(e_i,e_j)^{-1}.
\label{eq:s_genpairs}
\end{equation}
Consequently, the set
\begin{equation}
\{B_\omega(e_i,e_j):i<j\}
\label{eq:s_reconstructset}
\end{equation}
determines the complete cohomology class. At most $r(r-1)/2$ response settings are required, with factors having $\gcd(N_i,N_j)=1$ omitted.

For $G=\mathbb Z_N\times\mathbb Z_N$, take virtual Weyl operators $V_g=X$ and $V_h=Z$ obeying
\begin{equation}
XZ=e^{-2\pi\ii/N}ZX.
\end{equation}
Then
\begin{equation}
B_\omega(g,h)=e^{-2\pi\ii/N}.
\end{equation}
The Weyl-algebra construction applies for arbitrary $N$. For example, $G=\mathbb Z_3\times\mathbb Z_3$ has $H^2(G,U(1))\simeq\mathbb Z_3$, and the three values $B_\omega(g,h)=1,e^{\pm2\pi\ii/3}$ distinguish its three cohomology classes directly.

\section{Exact Gaussian limit and Gaussian-drop decoder}

The weak cluster measurements are Gaussian after a Jordan-Wigner transformation. We use the convention
\begin{equation}
X_i=\ii\gamma_{2i}\gamma_{2i+1},
\qquad
K_i=Z_{i-1}X_iZ_{i+1}
=\ii\gamma_{2i-1}\gamma_{2i+2}
\label{eq:s_JW}
\end{equation}
away from the Jordan-Wigner seam. Seam-crossing generators acquire the total fermion parity, so the decoder evaluates both parity sectors and recombines them exactly.

Let
\begin{equation}
\Gamma_{ab}=\langle \ii\gamma_a\gamma_b\rangle
\end{equation}
be a Majorana covariance matrix. Conditioning on
\begin{equation}
M_r(P)\propto e^{r\beta P/2},
\qquad
P=\ii\gamma_a\gamma_b,
\end{equation}
updates $\Gamma$ by a rank-two rational transformation. Writing $t=\tanh\beta$, $p=\Gamma_{ab}$, and $\alpha=rt/(1+rtp)$, pairs disjoint from $(a,b)$ receive
\begin{equation}
\Gamma'
=
\Gamma
+
\alpha
\left(
\Gamma_b\Gamma_a^{\mathsf T}
-
\Gamma_a\Gamma_b^{\mathsf T}
\right),
\label{eq:s_covupdate}
\end{equation}
with the measured pair and rows sharing one Majorana updated by the corresponding normalization and attenuation factors. This operation costs $O(L^2)$.

For a record of length $M$, each Gaussian branch propagates the identity effect backward; in normalized form this starts from the zero covariance matrix of the maximally mixed state. The even- and odd-sublattice charges and the total parity are Pfaffians of covariance submatrices, evaluated in $O(L^3)$ time. The complete cost is therefore
\begin{equation}
\text{time }O(ML^2+L^3),
\qquad
\text{memory }O(L^2).
\label{eq:s_decodercomplexity}
\end{equation}

\begin{proposition}[Exact trace-character decoder at $J=0$]
For any finite even $L\ge4$ and any finite record $s$ of the periodic monitored cluster instrument at $J=0$, the ordinary and canonically twisted $\mathbb Z_2$ trace moments
\begin{equation}
m_A^{(a)}(s)=
\frac{\Tr[U_AE_s^{(a)}]}{\Tr E_s^{(a)}},
\qquad a=0,h,
\label{eq:s_exactgaussianmoment}
\end{equation}
are obtained exactly by Gaussian covariance propagation together with exact projection onto the two Jordan--Wigner total-parity sectors. Consequently, away from a trace-character tie,
\begin{equation}
D_{A,a}^{\rm tr}(s)=\operatorname{sgn}m_A^{(a)}(s)
\label{eq:s_exactgaussiandecoder}
\end{equation}
is computable in time $O(ML^2+L^3)$ and memory $O(L^2)$, without a spectral-gap assumption.
\end{proposition}

\begin{proof}
Fix a branch $a$ and a record $s$. For each choice $q=\pm1$ of the total fermion parity $Q_{\rm tot}$, replace every seam-crossing cluster generator by its $q$-dependent Majorana-bilinear representative. Away from the Jordan--Wigner seam, the $X$-diagonal mapping gives
\begin{equation}
X_i=\ii\gamma_{2i}\gamma_{2i+1},
\qquad
K_i=\ii\gamma_{2i-1}\gamma_{2i+2},
\end{equation}
while a seam-crossing $K_i$ differs only by the fixed factor $q$. Hence every weak-measurement filter $M_r(P)\propto e^{r\beta P/2}$ is Gaussian in each $q$ representation. The canonical twist only reverses the sign of the designated seam outcome, so it preserves the same Gaussian structure. Backward covariance propagation therefore evaluates the corresponding unprojected Gaussian traces exactly.

For fixed $a$ and $s$, let $w_q$ denote the unprojected Gaussian trace weight in the $q$ representation, and let $m_{A,q}$, $m_{B,q}$, and $m_{{\rm tot},q}$ denote the normalized Pfaffian moments of $U_A$, $U_B$, and $Q_{\rm tot}=U_AU_B$ in that representation. The physical trace in parity sector $q$ is imposed exactly by
\begin{equation}
P_q=\frac{1+qQ_{\rm tot}}{2}.
\end{equation}
Using $U_AQ_{\rm tot}=U_B$ and summing the two projected sectors gives
\begin{equation}
m_A^{(a)}=
\frac{\sum_{q=\pm1}w_q\left(m_{A,q}+q m_{B,q}\right)}
{\sum_{q=\pm1}w_q\left(1+q m_{{\rm tot},q}\right)}.
\label{eq:s_parityrecombination}
\end{equation}
The common factors of $1/2$ from $P_q$ cancel between numerator and denominator. All moments in Eq.~\eqref{eq:s_parityrecombination} are Pfaffians of submatrices of the propagated $2L\times2L$ covariance matrix. Each of the $M$ filter updates costs $O(L^2)$, the final Pfaffians cost $O(L^3)$, and the two parity representations contribute only a constant factor. This proves both exactness and the stated complexity.
\end{proof}

As a direct finite-size check, we compared Eq.~\eqref{eq:s_exactgaussianmoment} with explicit Hilbert-space effects $E_s=K_s^\dagger K_s$ for random $J=0$ records at $\beta=2$ and record length $M=\max(8,3L)$. For 20 records each at $L=4$ and $L=6$, and five records at $L=8$, the maximum absolute ordinary/twisted discrepancies were respectively $(2.49\times10^{-14},6.22\times10^{-15})$, $(5.66\times10^{-15},5.55\times10^{-15})$, and $(5.55\times10^{-16},5.55\times10^{-16})$. The agreement is at floating-point precision.

The exact $J=0$ result isolates where approximation enters for interacting trajectories. The physical Born records used in the main text include the non-Gaussian interactions $e^{-\ii JX_iX_{i+2}}$. For these records, the same covariance/Pfaffian construction defines the Gaussian-drop decoder: it propagates the measured Gaussian filters and omits only the interaction gates from decoder propagation. The same polynomial computational structure is retained, while the interacting calculations below quantify the approximation's reconstruction performance beyond the Gaussian manifold.

\FloatBarrier
\begin{table}[b]
\caption{\textbf{Finite-size Gaussian-drop reconstruction.}
One hundred Born records are used for each phase and size in the practical response data.}
\label{tab:s_gaussianscaling}
\begin{ruledtabular}
\begin{tabular}{cccccc}
$L$ & sweeps & phase & $\widetilde W$ & standard error & decoder failures\\
\hline
8 & 8 & SPT & $-0.86$ & 0.051 & 7\%\\
8 & 8 & trivial & $+0.90$ & 0.044 & 5\%\\
10 & 12 & SPT & $-0.96$ & 0.028 & 2\%\\
10 & 12 & trivial & $+0.96$ & 0.028 & 2\%\\
12 & 16 & SPT & $-0.96$ & 0.028 & 2\%\\
12 & 16 & trivial & $+1.00$ & 0.000 & 0\%
\end{tabular}
\end{ruledtabular}
\end{table}
\FloatBarrier

Two additional scans resolve how the practical reconstruction evolves across the measurement-controlled crossover and with monitoring time. Figure~\ref{fig:s_extendedscans}(a) uses 100 records at each of seven $p_X$ values for $L=8,10,12$ with $T=2L-8$; the response approaches the quantized signs away from the transition and crosses through a broad finite-size region near $p_X\simeq 0.5$. Figure~\ref{fig:s_extendedscans}(b) fixes $L=8$ at the two deep-phase cuts and uses 100 records at each monitoring time, showing progressive sharpening toward the corresponding phase value as $T$ increases. Together these scans extend the practical-decoder validation across phase and time parameters.

\begin{figure}[t]
\centering
\includegraphics[width=0.98\linewidth]{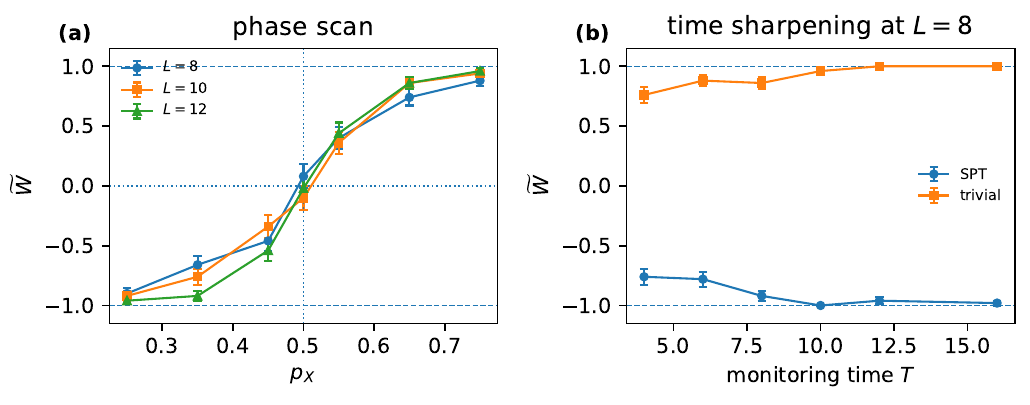}
\caption{\textbf{Extended practical-decoder validation.}
(a) Reconstructed response versus measurement-selection probability $p_X$ for three sizes, with $T=2L-8$ and 100 interacting records per point. Error bars are standard errors of the bounded single-record estimator. (b) Monitoring-time scan at $L=8$ for the deep SPT and trivial cuts, again with 100 interacting records per point. Dashed lines mark the quantized values.}
\label{fig:s_extendedscans}
\end{figure}

Because the endpoint-mismatch and decoder-error indicators are binomial events, Table~\ref{tab:s_binomial_ci} reports exact two-sided 95\% Clopper--Pearson intervals. For zero observed events, the interval gives the corresponding finite upper confidence bound.

\begin{table}[t]
\caption{\textbf{Finite-sample confidence intervals.}
Counts are followed by exact two-sided 95\% Clopper--Pearson intervals for the underlying event probability. The endpoint-mismatch counts come from independent full interacting-effect calculations; they include a readout-character gap-closing tie under the fixed tie-breaking convention. The decoder-failure counts come from the 100-record practical reconstruction data.}
\label{tab:s_binomial_ci}
\begin{ruledtabular}
\begin{tabular}{cccccc}
$L$ & phase & endpoint mismatch & 95\% CI & decoder failure & 95\% CI\\
\hline
8 & SPT & $3/100$ & $[0.62,8.52]\%$ & $7/100$ & $[2.86,13.89]\%$\\
8 & trivial & $1/100$ & $[0.03,5.45]\%$ & $5/100$ & $[1.64,11.28]\%$\\
10 & SPT & $0/50$ & $[0,7.11]\%$ & $2/100$ & $[0.24,7.04]\%$\\
10 & trivial & $0/50$ & $[0,7.11]\%$ & $2/100$ & $[0.24,7.04]\%$\\
12 & SPT & $0/20$ & $[0,16.84]\%$ & $2/100$ & $[0.24,7.04]\%$\\
12 & trivial & $0/20$ & $[0,16.84]\%$ & $0/100$ & $[0,3.62]\%$
\end{tabular}
\end{ruledtabular}
\end{table}

The implementation benchmark uses synthetic records of length $M=L(2L-8)$. Median wall-clock time grows from $2.47$ ms at $L=8$ to $0.203$ s at $L=32$, with an empirical power-law exponent near $3.14$ on the tested range. Equation~\eqref{eq:s_decodercomplexity} gives the analytical complexity, while the empirical fit characterizes the tested implementation.

\section{Numerical methods and reproducibility}

The full interacting-effect calculations use the matrix-free maps
\begin{equation}
v\mapsto K v,\qquad
v\mapsto K^\dagger v,\qquad
v\mapsto E v,
\end{equation}
together with Lanczos or power iteration inside explicit symmetry sectors. The memory cost is $O(2^L)$. For the joint $(U_A,U_B)$ calculation, each sector has dimension $2^{L-2}$.

The path operators use the scaled filters
\begin{equation}
M_r(P)/\|M_r(P)\|
\end{equation}
to prevent numerical overflow. This rescaling removes record- and path-dependent global positive scalars while leaving every eigenvalue ratio, charge label, and normalized eigenvector unchanged. The two largest eigenvalues are computed in each complete sector. Residuals are reported as
\begin{equation}
\frac{\|Ev_j-\lambda_jv_j\|}{\lambda_1}.
\label{eq:s_residual}
\end{equation}
Across the 20-record-per-phase reproduction, the largest reported residual remains below $10^{-9}$, while typical residuals are near machine precision.

Table~\ref{tab:s_configs} lists the principal configurations and random seeds. The reproducibility code is included as ancillary material with this arXiv submission and regenerates the reported numerical data from these specified parameters and seeds.

\begin{table}[t]
\caption{\textbf{Principal numerical configurations.}}
\label{tab:s_configs}
\begin{ruledtabular}
\begin{tabular}{lcccccc}
task & $L$ & sweeps & records & $p_X$ & seed & solver\\
\hline
joint paths & 8 & 8 & 20/phase & 0.25,0.75 & 202607192/194 & Lanczos\\
selected crossing & 8 & 8 & 1 & 0.25 & 202607181 & Lanczos\\
critical control & 8 & 8 & 4 & 0.50 & 202607193 & Lanczos\\
full-effect calculation & 8 & 8 & 100/phase & 0.25,0.75 & 271828 & power\\
full-effect calculation & 10 & 12 & 50/phase & 0.25,0.75 & 271828 & power\\
full-effect calculation & 12 & 16 & 20/phase & 0.25,0.75 & 13579 & power\\
practical decoder & 8,10,12 & 8,12,16 & 100/phase & 0.25,0.75 & 424242 & Gaussian\\
runtime benchmark & 8--32 & $2L-8$ & 3--5 repeats & synthetic & --- & Gaussian
\end{tabular}
\end{ruledtabular}
\end{table}

The physical parameters are
\begin{equation}
\beta=2,\qquad J=0.1,\qquad r_J=0.2.
\end{equation}
The deep SPT and trivial cuts are $p_X=0.25$ and $0.75$. The full path validation uses nine equally spaced $\lambda$ points, followed by a 33-point and continuous minimization near the weakest apparent gap. The critical control at $p_X=0.50$ displays mixed endpoint labels and reduced readout-character gaps, as expected.

For the binary response in Fig.~2(c) of the main text, error bars are exact two-sided 95\% Clopper--Pearson intervals obtained from the observed failure counts. Standard errors in Table~\ref{tab:s_gaussianscaling} are descriptive sample uncertainties. Zero observed events retain the finite upper confidence bounds reported in Table~\ref{tab:s_binomial_ci}.

\section{Experimental realization and symmetry-breaking noise}

All quantum operations in the protocol have standard ancilla-based realizations. A weak Pauli-string measurement can be implemented by preparing an ancilla in $|0\rangle_a$, applying
\begin{equation}
U(\theta)=e^{-\ii\theta P\otimes Y_a},
\end{equation}
and measuring the ancilla in the $X$ basis. The conditional data Kraus operator is
\begin{equation}
\langle r_X|U(\theta)|0\rangle_a
=
\frac{1}{\sqrt2}
\left(
\cos\theta\,I+r\sin\theta\,P
\right),
\end{equation}
which realizes Eq.~\eqref{eq:s_weakfilter} after setting
\begin{equation}
\tan\theta=\tanh(\beta/2).
\end{equation}
The weight-three operator $Z_{i-1}X_iZ_{i+1}$ can be measured by basis rotation on the central qubit, mapping the Pauli parity to an ancilla, applying the partial interaction, and reversing the mapping.

The randomized input is an $X$-basis product state with known sublattice charges. The experimental layer consists of ordinary-boundary monitoring and record storage, with all twist processing performed offline.

Exact quantization follows under record-wise strong symmetry. For calibrated deviations from that ideal, the finite-record estimator obeys a direct quantitative stability bound. Let $P(n,s)$ be the ideal joint distribution of randomized input labels and ordinary records, let $\widetilde P(n,s)$ be the noisy distribution, and use the same unit-modulus post-processing variable $X(n,s)=\chi_n(g)^*\widetilde D_{g,h}(s)$. Then
\begin{equation}
\left|\mathbb E_{\widetilde P}X-\mathbb E_P X\right|
\le 2\,d_{\rm TV}(\widetilde P,P),
\label{eq:s_noise_tv}
\end{equation}
where $d_{\rm TV}=\frac12\sum_{n,s}|\widetilde P-P|$. If each of the $M$ monitored steps, including its classical outcome register, differs from the calibrated symmetric instrument by at most $\delta$ in one-half diamond norm, the standard telescoping bound gives $d_{\rm TV}\le M\delta$ and hence a worst-case response bias at most $2M\delta$. This expresses the response stability directly in terms of experimentally calibratable instrument deviations.

Relevant symmetry-breaking mechanisms include amplitude damping, asymmetric coherent over-rotations, residual couplings, readout crosstalk, leakage, and imperfect ancilla reset. A corresponding symmetry-leakage diagnostic is, for example
\begin{equation}
\epsilon_{\rm sym}(s,g)
=
\frac{\|[E_s,U_g]\|}{2\|E_s\|},
\end{equation}
or an experimentally accessible proxy obtained from symmetry-twirled calibration data. Symmetry-preserving compilation and randomized symmetry twirling can suppress coherent violations. A noise-aware decoder can incorporate calibrated readout confusion matrices and effective Kraus maps.

The principal experimental requirements are repeated high-fidelity weight-three weak measurements, ancilla reset over many monitoring rounds, and preservation of the two sublattice symmetries. Programmable trapped-ion and superconducting architectures provide routes to small-system proof-of-principle demonstrations; extending the protocol to thermodynamic scaling requires correspondingly longer high-fidelity monitoring sequences.

\end{document}